\documentclass[pra,aps,preprint,showpacs,preprintnumbers,
amsmath,amssymb,floatfix]{revtex4-1}
\usepackage{enumerate}
\usepackage{graphicx}
\usepackage{amsmath, amsthm, amssymb,bm}
\newtheorem{theorem}{Theorem}

\newtheorem{lemma}{Lemma}

\newtheorem{remark}{Remark}

\begin{document}
\newcommand{\real}{\textrm{Re}\:}
\newcommand{\sto}{\stackrel{s}{\to}}
\newcommand{\supp}{\textrm{supp}\:}
\newcommand{\wto}{\stackrel{w}{\to}}
\newcommand{\ssto}{\stackrel{s}{\to}}
\newcounter{foo}
\providecommand{\norm}[1]{\lVert#1\rVert}
\providecommand{\abs}[1]{\lvert#1\rvert}

\title{Universal low-energy behavior in three-body systems}

\author{Dmitry K. Gridnev}
\affiliation{FIAS, Ruth-Moufang-Stra{\ss}e 1, D--60438 Frankfurt am Main,
Germany}
\altaffiliation[On leave from:  ]{ Institute of Physics, St. Petersburg State
University, Ulyanovskaya 1, 198504 Russia}

\begin{abstract}
We consider a pairwise interacting quantum 3-body system in 3-dimensional space
with finite masses and the interaction term $V_{12} + \lambda(V_{13} + V_{23})$, where all pair potentials are assumed to be nonpositive.
The pair interaction of the particles $\{1,2\}$ is tuned to make them have a zero energy resonance 
and no negative energy bound states. The
coupling constant $\lambda >0$ is allowed to take the values for which the particle pairs
$\{1,3\}$ and $\{2,3\}$ have 
no bound states with negative energy. Let $\lambda_{cr}$ denote the critical
value of the coupling constant such that $E(\lambda) \to -0$ for  
$\lambda \to \lambda_{cr}$, where $E(\lambda)$ is the ground state energy of the
3-body system. We prove the theorem, which states that near $\lambda_{cr}$ one has $E(\lambda)
= C (\lambda-\lambda_{cr})[\ln (\lambda-\lambda_{cr})]^{-1}+$h.t., 
where $C$ is a constant and h.t. stands for ``higher terms''. This behavior of the
ground state energy is universal (up to the value of the constant $C$), meaning that it is independent of the form of pair
interactions. 
\end{abstract}

\maketitle


\section{Introduction}\label{sec:1}

Universality plays an important role in physics. The interest to it is inspired by the striking similarity 
in behavior near the critical point among systems
that are otherwise quite different in nature. For example, various substances, which exhibit liquid-gas phase transition, near the critical point 
obey the universal law $\rho_{gas} -\rho_c \to -A(T_c - T)^\beta$. Here $\rho_{gas}, \rho_c$ denote the density of gas and critical density respectively, $T,T_c$ are 
temperature and critical temperature, $A$ is a constant and $\beta$ is the so-called critical exponent \cite{hammer}. Amazingly, the value of $\beta \simeq 0.325$ is the same 
for many substances, which are completely different on the atomic level. Similar law with the same value of the critical exponent holds true for magnetization 
in ferromagnets as a function of temperature. Another example of universality is found in the ground state energy of the Bose gas  as a function of density. In the low density limit it 
approaches an expression, which depends only on the scattering length but not on the overall form of pair interaction \cite{yngv}.  

Small quantum systems also exhibit universal features \cite{hammer}. 
One example of universality in the two-particle case concerns the behavior of
the energy depending on the coupling constant near the threshold. Suppose that 
$E(\lambda)$ is the energy of an isolated non-degenerate state of the Hamiltonian $h(\lambda) = T + \lambda
V_{12}$ in 3-dimensional space and $E(\lambda) \to 0$ for $\lambda \to
\lambda_{cr}$. 
Then universally for $\lambda$ near $\lambda_{cr}$ one has $E(\lambda) = 
c(\lambda - \lambda_{cr})^2 +$\rm{h.t.} or $E(\lambda) \simeq c(\lambda - \lambda_{cr})+$\rm{h.t.} 
depending 
on whether zero is an eigenvalue of $h(\lambda_{cr})$ or not, see \cite{klaus1}. \textbf{Universal} in this context means 
that this behavior up to a constant is true for all short range interactions independently of their form.  
``h.t.'' is the shorthand notation for ``higher terms'' and $E(\lambda) = f(\lambda) +$\rm{h.t.} for $f(\lambda) \to 0$ always
implies that $E(\lambda) = f(\lambda) + \hbox{o}\bigl(f(\lambda)\bigr)$. 

If two particles are set into an 
n-dimensional space 
the scenario depends on the space dimension \cite{klaus1}: for example, in 2-dimensional flatland 
the energy of the ground state energy can approach zero exponentially fast 
$E(\lambda) =\exp (-c(\lambda-\lambda_{cr})^{-1})+$\rm{h.t.}, 
and in 4 dimensions the ground state energy approaches zero very slow, namely, 
$E(\lambda) = c(\lambda - \lambda_{cr})|\log (\lambda -
\lambda_{cr})|^{-1}+$\rm{h.t.}. For a full account of possible scenarios see Table~I in \cite{klaus1}. 
Another universality associated with 2-body system in 3-dimensional space relates to the 
wave function near the threshold. If the energy of a non-degenerate bound state near
the threshold satisfies $E(\lambda) = c(\lambda - \lambda_{cr})^2+$\rm{h.t.} then the wave
function 
of this bound state $\psi(\lambda, x)$ approaches spherically
symmetric expression \cite{schwinger} 
\begin{equation}\label{24.1:1}
 \left\| \psi(\lambda, x) - |E(\lambda)|^{1/4}\frac{e^{-|E(\lambda)|^{\frac 12}
|x|}}{\sqrt{2\pi}|x|}\right\| \to 0 , 
\end{equation}
see Eq.~(8) in \cite{3}, where we have omitted the phase factor. For
well-behaved short-range interactions Eq.~(\ref{24.1:1}) holds irrespectively of the form
of the pair potential.

In 3-particle systems the notorious example of universality is the Efimov
effect. Efimov's striking and counterintuitive prediction \cite{vefimov} was that just by tuning coupling constants of the short-range interactions in the 
3-body system one can bind an infinite number of levels, even though the two-body subsystems bind none. The infinitude of bound states was shown rigorously by Yafaev in \cite{yafaev}. 
Basing on the Yafaev's method Sobolev \cite{sobolev} has proved that 
\begin{equation}
 \lim_{\epsilon \to 0} |\ln \epsilon|^{-1}N(\epsilon) = \mathfrak{U}_0/2, 
\end{equation}
where $N(\epsilon)$ is the number of bound states with the energy less than $-\epsilon < 0$ and $\mathfrak{U}_0$ is the universal positive constant, 
which depends only on masses. Let us remark that in physics\cite{vefimov,hammer,helfrich} it is generally conjectured that 
\begin{equation}\label{1.14;10}
 \lim_{n \to \infty} \frac{|E_n|}{|E_{n+1}|} = e^{2\pi/s_0} , 
\end{equation}
where $E_n$ is the energy of the $n$-th Efimov level and $s_0 = \pi \mathfrak{U}_0$. In physics this is termed as universal scaling of Efimov levels, see \cite{hoegkrohn,lakaev} for the mathematical discussion of (\ref{1.14;10}). 
Efmiov's prediction was later confirmed experimentally in ultracold gases \cite{kraemer}. 
The so-called 4-body universality \cite{naturephysics} holds only approximately \cite{myfbs} and the question of finite range corrections 
is still being debated \cite{debate}.

For further discussion it is useful to introduce the following mathematical notations. 
For an operator $A$ acting on a Hilbert space $D(A)$, $\sigma(A)$ and $\sigma_{ess} (A)$ denote the domain, the spectrum, and 
the essential spectrum of $A$ respectively \cite{reed}. $A > 0$ means that $(f, Af) > 0$ for all $f \in D(A)$, while 
$A \ngeq 0$ means that there exists $f_0 \in D(A)$ such that $(f_0, Af_0) < 0$. 
$\mathfrak{B}(\mathcal{H})$ denotes the set of bounded linear operators on the Hilbert space $\mathcal{H}$.  For an interval
$\Omega \subset \mathbb{R}$ the function 
$\chi_\Omega : \mathbb{R} \to \mathbb{R}$ is such that $\chi_\Omega (x) = 1$ if
$x\in \Omega$ and $\chi_\Omega (x) = 0$ otherwise.

Recently a new type of universality in 3-body systems has been discovered in \cite{3}. 
Consider the Hamiltonian of the 3--particle system in $\mathbb{R}^3$
\begin{equation}\label{hami}
 H(\lambda) = H_0 + v_{12} + \lambda (v_{13} + v_{23}) ,
\end{equation}
where $H_0$ is the kinetic energy operator with the center of mass removed,
$\lambda >0$ is the coupling constant and none of the particle
pairs has negative energy bound states. All particles are supposed to have a finite mass. The pair--interactions $v_{ik}$ are operators of multiplication
by real $V_{ik} (r_i - r_k) \leq 0$ and $r_i \in \mathbb{R}^3$ are particle position vectors. For pair potentials we require like in \cite{3} that 
\begin{equation}
\gamma_0 := \max_{i=1,2}\max\left[ \int d^3 r \bigl| V_{i3} (r)\bigr|^2 , \int
d^3 r \bigl| V_{i3} (r)\bigr| (1+|r|)^{2\delta}\right] < \infty , \label{restr}
\end{equation}
where $0 < \delta < 1/8$ is a fixed constant, and
\begin{equation}\label{restr3}
 -b_1 e^{-b_2 |r|} \leq V_{12} (r) \leq 0 ,
\end{equation}
where $b_{1,2} >0$ are constants.

We shall assume that the interaction between 
the particles $\{1,2\}$ is tuned to make them have a zero energy resonance 
and no negative energy bound states \cite{yafaev2}. This implies that in the absence of particle 3 the particles $\{1,2\}$ are ``almost'' bound, that is a bound state 
with negative energy appears if and only if the interaction $v_{12}$ is strengthened by a negligibly small amount. In mathematical terms this can be expressed as follows 
\begin{gather}
 H_0 + v_{12} >0 \label{14.1;2}\\
H_0 + (1+\varepsilon)v_{12} \ngeq 0 \quad \quad \textnormal{for all $\varepsilon >0$.} \label{14.1;3}
\end{gather}
Let $\lambda'_{1,2}$ be the values of the coupling constants such that $H_0 + \lambda'_1 v_{13}$ 
and $H_0 + \lambda'_2 v_{23}$ are at critical coupling in the sense of Def.~1 in \cite{2}. This means that $H_0 + \lambda'_i v_{i3} >0$ 
and $H_0 + (1+\varepsilon)\lambda'_i v_{i3} \ngeq 0$ for all $\varepsilon >0$ and $i=1,2$. Let us set 
\begin{equation}
 \tilde \lambda := \min [\lambda'_1 , \lambda'_2].  \label{tildelambda}
\end{equation}
We shall always assume that the coupling constant in (\ref{hami}) satisfies the inequality $\lambda < \tilde \lambda$. In other words, the coupling constant takes the values 
for which the particle pairs $\{1,3\}$ and $\{2,3\}$ have neither zero energy resonances nor bound states with negative energy. Under these conditions 
$\sigma_{ess} (H(\lambda)) =[0, \infty)$ and $H(\lambda)$ has a finite number of bound states with negative energy as proved in \cite{yafaevnew}. 

Let 
$\lambda_{cr}$ be the value of the coupling constant such that $H(\lambda_{cr}) \geq 0$ but $H(\lambda_{cr} + \epsilon) \ngeq 0$ for all $\epsilon >0$. 
In Sec. 6 in \cite{1} it is proved that 
$\lambda_{cr} \in (0, \tilde \lambda)$. 
By the HVZ theorem (see \cite{reed}, Vol. 4 and \cite{teschl}) 
for $\lambda \in (\lambda_{cr}, {\tilde \lambda})$ 
\begin{equation}
 H(\lambda) \psi_\lambda = E(\lambda) \psi_\lambda , 
\end{equation}
where $E(\lambda) <0$ is the ground state energy and $\psi_\lambda \in D(H_0)$. By Theorems~1, 3 in 
\cite{1} and Theorem~2 in \cite{3} $E(\lambda) \nearrow 0$ for $\lambda \searrow \lambda_{cr}$ but zero is not an eigenvalue of $H(\lambda_{cr})$. Moreover, 
$\psi_\lambda$ for $\lambda \searrow \lambda_{cr}$ totally spreads, that is 
\begin{equation}\label{14.1;1}
 \lim_{\lambda\to \lambda_{cr}} \int_{|x^2| + |y^2| < R} |\psi_\lambda (x,y)|^2 d^3 x d^3 y= 0  \quad \quad \textnormal{for all $R>0$} . 
\end{equation}
In (\ref{14.1;1}) $x,y \in \mathbb{R}^3$ are Jacobi coordinates in the 3-body problem, which are shown in Fig.~\ref{fig:1} (left). They are defined as 
\begin{gather}
 x =\alpha^{-1} (r_2 - r_1 ) \nonumber, \\
y = \frac{\sqrt{M_{12}}}{\hbar} \left[ r_3 - \frac{m_1}{(m_1 + m_2)} r_1 - \frac{m_2}{(m_1 + m_2)} r_2\right] \nonumber , 
\end{gather}
where $m_i$ denote particle masses, $\alpha = \hbar (m_1 + m_2)^{\frac 12} (2m_1 m_2)^{-\frac 12}$ and 
$M_{12} = (m_1 + m_2) m_3 /(m_1 + m_2 + m_3)$.

In \cite{3} it was proved that $\psi_\lambda$ for $\lambda\to \lambda_{cr}$ 
approaches in norm a universal expression, namely, 
\begin{equation}\label{3bh3}
  \psi_\lambda \to \frac{1}{\sqrt 2 \pi^{3/2} |\ln
|E(\lambda)||^{1/2}} \frac{ \bigl\{ |x|\sin(k_n |y|) + |y|\cos(k_n |y|)\bigr\}\exp(-|E(\lambda)|^{1/2}|x|)}{1+|x|^3|y|+|y|^3 |x|} . 
\end{equation}
In the limit the wave function $\psi_\lambda$ describes the state, in which average distances between all three particles go to infinity. 
(This is partly the reason why the short range details of pair interactions becomes unimportant). 

By analogy with the 2-particle case 
it is natural to assume that $E(\lambda)$ would exhibit universal behavior near $\lambda_{cr}$. In this paper we shall prove Theorem~\ref{th:main}, which states 
that it is indeed so and universally 
one has $E(\lambda)
= C (\lambda-\lambda_{cr})[\ln (\lambda-\lambda_{cr})]^{-1} +$\rm{h.t.}, where $C >0$. 
It is important to stress that this result does 
not follow directly from (\ref{3bh3}), if one merely tries to substitute (\ref{3bh3}) into the equation $E(\lambda) = \bigl(\psi_\lambda , H(\lambda) \psi_\lambda \bigr)$.  
This is due to the error terms, which in spite of going to zero in norm may affect the 
resulting average, for a detailed explanation see Remark~\ref{newremark} in the next section. The obtained behavior of $E(\lambda)$ remarkably mimics that of the ground state 
energy of 2 particles in 4-dimensional space. The experimental 
observation of this type of universality can possibly be obtained in ultracold gas mixtures, see \cite{3} for discussion. 
When the pair interaction $v_{12}$ is not tuned there are 3 
types of possible asymptotic behavior of $E(\lambda)$, which are listed in Theorem~\ref{th:main1a}.

\begin{figure}
\includegraphics[height=0.16\textheight]{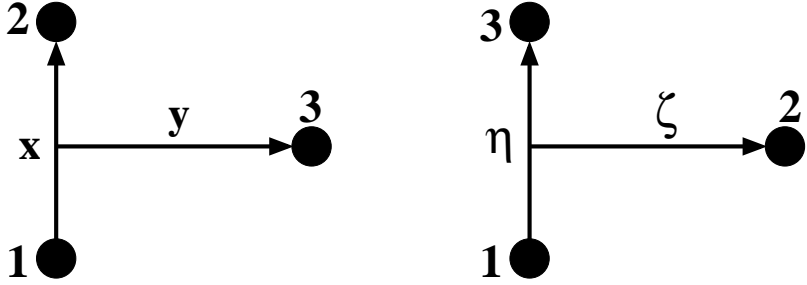}
\caption{Jacobi coordinates in the three-particle problem. $y$ points in the direction from the center of mass of the particles $\{1,2\}$ to particle $3$.
 $\zeta$ points in the direction from the center of mass of the particles $\{1,3\}$ to particle $2$.  The picture shows only directions of
the vectors, the scales are set in order to ensure that $H_0 =
-\Delta_{x}-\Delta_{y}$ and $H_0 =
-\Delta_{\eta}-\Delta_{\zeta}$ holds.}
\label{fig:1}
\end{figure}

\section{Main Result}\label{sec:2}

The Hamiltonian $H(\lambda)$ in (\ref{hami}) is self--adjoint on 
$D(H_0) = \mathcal{H}^2 (\mathbb{R}^{6}) \subset L^2 ( \mathbb{R}^{6})$, where 
 $\mathcal{H}^2 (\mathbb{R}^{6})$ denotes the 
corresponding Sobolev space \cite{teschl,liebloss}. 
The pair interaction between particles $\{1,2\}$ is tuned 
so that they have a zero energy resonance, that is Eqs.~(\ref{14.1;2})-(\ref{14.1;3}) hold.

Our aim in this paper is to prove 
\begin{theorem}\label{th:main}
Suppose that $E(\lambda):= \inf \sigma(H(\lambda))$, then for $\lambda \searrow \lambda_{cr}$ one has
\begin{equation}\label{main}
 E (\lambda) =  C_0 \frac{(\lambda- \lambda_{cr})}{\ln (\lambda-\lambda_{cr})} + \textnormal{\hbox{o}} \left( \frac{(\lambda- \lambda_{cr})}{\left| \ln (\lambda-\lambda_{cr})\right|}\right) , 
\end{equation}
where $C_0 > 0$ is a finite constant. 
\end{theorem}
Before we proceed with the proof let us remark that 1) Eq.~(\ref{main}) is universal, meaning that up to a constant it does not depend on the details of pair interaction; 
2) the function $E(\lambda)$ at $\lambda=\lambda_{cr}$ cannot be Taylor expanded in powers of $(\lambda-\lambda_{cr})^\alpha$ for any $\alpha >0$; 
3) the three-body ground state energy in the 3-dimensional case has the same behavior near $\lambda_{cr}$ as the 2-body ground state energy in the 4-dimensional case 
(we do not have an explanation for this finding); 4) the constant $C_0$ depends on pair interactions and can be expressed through zero energy solutions of Birman-Schwinger operators; 
5) the method of the proof is different from \cite{klaus1,klaus2}: the method in \cite{klaus1,klaus2}, which uses the low energy expansions of the Birman-Schwinger operator, is 
not applicable here. The proof below hinges on Theorem~\ref{th:main2} in Section~\ref{sec:3}, whose technical proof is heavily based on the results and methods in \cite{3}.

\begin{proof}[Proof of Theorem~\ref{th:main}]
 For $\lambda \in (\lambda_{cr}, {\tilde \lambda} )$ there exists $\psi_\lambda \in
D(H_0)$, $\|\psi_\lambda\| = 1$  such that $H(\lambda) \psi_\lambda = E
(\lambda)\psi_\lambda$, 
besides we can assume that $\psi_\lambda >0$ because it is the ground state. $E (\lambda)$ is
smooth and monotone increasing on $(\lambda_{cr}, {\tilde \lambda})$. Using
perturbation theory \cite{reed,kato} we obtain 
\begin{equation}\label{24.1:9}
 -\frac{dE (\lambda)}{d\lambda} = \||v_{13}|^{1/2} \psi_\lambda\|^2 +
\||v_{23}|^{1/2} \psi_\lambda\|^2 . 
\end{equation}
 By Theorem~\ref{th:main2} for $\lambda$ close enough to $\lambda_{cr}$ there exists a constant $C_0 >0$ such that 
\begin{equation}
 -\frac{C_0+\varepsilon }{\ln (-E (\lambda))}  \geq -\frac{dE (\lambda)}{d\lambda} \geq -\frac{C_0-\varepsilon  }{\ln (-E (\lambda))} 
\end{equation}
for any given $\varepsilon >0$. 
The last inequality can be equivalently rewritten as 
\begin{equation}\label{24.1:3}
 C_0 + \varepsilon  \geq \frac{d}{d\lambda} \left( E(\lambda) \ln (-E(\lambda)) - E(\lambda)\right) \geq C_0 - \varepsilon . 
\end{equation}
Integrating (\ref{24.1:3}) we obtain 
\begin{equation}\label{24.1:4}
  (C_0 + \varepsilon) (\lambda-
\lambda_{cr}) \geq E (\lambda) \ln (-E (\lambda)) - E(\lambda) \geq (C_0 - \varepsilon) (\lambda-
\lambda_{cr}) . 
\end{equation}
Let us set 
\begin{equation}\label{24.1:6}
 E (\lambda) =-f(\lambda)(\lambda - \lambda_{cr}) ,  
\end{equation}
where $f (\lambda) >0$. From (\ref{24.1:4}) we get 
\begin{equation}\label{24.1:5}
  (C_0 + \varepsilon) \geq f (\lambda)[- \ln (-E (\lambda)) +1 ] \geq (C_0 - \varepsilon) . 
\end{equation} 
Using that $E(\lambda) \to 0$ from (\ref{24.1:5}) we conclude that 
$\lim_{\lambda \to \lambda_{cr}} f(\lambda ) = 0$. Again substituting (\ref{24.1:6}) into (\ref{24.1:5}) we obtain 
\begin{equation}
  (C_0 + \varepsilon) \geq -f (\lambda)\ln (\lambda - \lambda_{cr}) - [ f(\lambda) \ln (f(\lambda))  -f(\lambda) ] \geq (C_0 - \varepsilon) 
\end{equation} 
The term in square brackets in the last inequality goes to zero for $\lambda \to \lambda_{cr}$. Thus for $\lambda$ close to $\lambda_{cr}$ we have 
\begin{equation}\label{24.1:8}
  - \frac{C_0 + \varepsilon/2 }{\ln (\lambda - \lambda_{cr})} \geq f (\lambda) \geq - \frac{C_0 - \varepsilon/2 }{\ln (\lambda - \lambda_{cr})} . 
\end{equation} 
Now (\ref{24.1:9}) follows from (\ref{24.1:6}), (\ref{24.1:8}) since $\varepsilon >0 $ is arbitrarily small. \end{proof}
\begin{remark}\label{newremark}
 We should explicitly warn against the direct ``physicist's approach'', when one substitutes in (\ref{24.1:9}) instead of $\psi_\lambda$ the rhs of (\ref{3bh3}). 
Eq.~(\ref{3bh3}) defines $\psi_\lambda$ up to error terms, which go to zero in norm when the energy goes to zero. There is otherwise no control of these error 
terms and one cannot exclude the situation, where, for example, 
\begin{equation}\label{yasego}
 \psi_\lambda = \Phi_\lambda + \bigl|\ln |E(\lambda)|\bigr|^{-\frac 18} \exp(-|x|^2 -|y|^2)
\end{equation}
and $\Phi_\lambda$ equals the rhs of (\ref{3bh3}). If one substitutes (\ref{yasego}) into (\ref{24.1:9}) one would find that the resulting 
$E(\lambda)$ would be very different from the form stated in 
Theorem~\ref{th:main}. 
\end{remark}

Theorem~\ref{th:main} considers the case when the particles $\{1,2\}$ have a zero energy resonance. Now let us consider a more general situation and 
assume that the 3-particle system is described 
by the Hamiltonian (\ref{hami}), where for simplicity we require that $V_{ik} \leq 0$ are bounded and have a compact support. 
We still require that $\lambda < {\tilde \lambda}$, {\it i.e.} $\lambda$ takes the values for 
which the subsystems $\{1,3\}$ and $\{2,3\}$ have no bound states with negative energy and no zero energy resonances. However, we do not impose restrictions on the 
spectrum of the particles $\{1,2\}$, which means that this pair determines the energy of the dissociation threshold $E_{thr}$, that is 
\begin{equation}
 E_{thr} := \inf \sigma_{ess} (H(\lambda)) = \inf \sigma (H_0 + v_{12}) . 
\end{equation}
The critical coupling constant $\lambda_{cr}$ is the value of $\lambda$ for which the 3-body bound state, 
whose energy lies below $E_{thr}$, is about to be formed. Mathematically speaking
\begin{equation}
 \lambda_{cr} = \sup\{\lambda | \inf \sigma (H(\lambda)) = E_{thr}\} . 
\end{equation}
By the methods similar to \cite{oldstuff} one can prove that $\lambda_{cr} < {\tilde \lambda}$. 
\begin{theorem}\label{th:main1a}
One can distinguish 3 cases: (A) the pair $\{1,2\}$ has no negative energy bound states and no zero energy resonance; (B) the pair $\{1,2\}$ has no negative energy bound states 
but has a zero energy resonance; (C) the pair $\{1,2\}$ has at least one bound state with negative energy. Suppose that 
$E(\lambda):= \inf \sigma(H(\lambda))$, then for $\lambda \searrow \lambda_{cr}$ in each case one has
\begin{align}
&(A) \quad \quad E(\lambda) - E_{thr}=E(\lambda)= - c (\lambda-\lambda_{cr}) + \mbox{\textnormal{h.t.}} \nonumber\\
&(B) \quad \quad  E(\lambda) - E_{thr}=E(\lambda)= c (\lambda-\lambda_{cr})[\ln (\lambda-\lambda_{cr})]^{-1}+ \mbox{\textnormal{h.t.}} \nonumber\\
&(C) \quad \quad  E(\lambda)-E_{thr} = - c (\lambda-\lambda_{cr})^2 + \mbox{\textnormal{h.t.}} \nonumber
\end{align}
where $c > 0$ is a finite constant.  
\end{theorem}
\begin{proof}
 Case (B) follows from Theorem~\ref{th:main} and case (C) was proved in Theorem~3.2 in \cite{klaus2}. In case (A) let $\psi_\lambda$ denote the eigenfunction of $H(\lambda)$, which corresponds to 
the eigenvalue $E(\lambda)$. As follows from the proof of Theorem 2 in \cite{1} $\psi_\lambda \to \psi_0$ in norm, where $\psi_0 \in D(H_0)$ is the eigenfunction corresponding to the 
zero eigenvalue of $H(\lambda_{cr})$. Thus $\||v_{13}|^{1/2}\psi_\lambda\|^2 + \||v_{23}|^{1/2}\psi_\lambda\|^2 \to c$, where $c \in (0,\infty)$ and, hence, 
 $dE/d\lambda \to -c$. The rest of the proof is trivial. 
\end{proof}

\section{Asymptotics of Potential Energy Terms}\label{sec:3}

Let $e(\mu)$ be the ground state energy of the Hamiltonian $h(\mu) = -\Delta_x + \mu v_{12} $ acting on $L^2 (\mathbb{R}^3)$, where $v_{12}$ is the operator of 
multiplication by $V_{12} (\alpha x)$ and $\mu \in \mathbb{R}_+$ is the coupling constant. 
Because the pair of particles $\{1,2\}$ has a 
zero energy resonance and no negative energy bound states the following expansion \cite{klaus1}
\begin{equation}
 e(\mu) = a^{-2} (\mu -1)^2 + \hbox{o} (\mu-1)
\end{equation}
 is true for $\mu \to + 1$. The positive constant $a$ in this expansion 
can be expressed through the zero energy solution of the Birman-Schwinger operator, namely \cite{klaus1}, 
\begin{equation}
 a = (4\pi )^{-1} \bigl\| |v_{12}|^{\frac 12} \phi_0\bigr\|_1^2 , 
\end{equation}
where $\phi_0 \in L^2 (\mathbb{R}^3)$ is the unique nonnegative and normalized solution of the equation 
\begin{equation}\label{1.16;1}
|v_{12}|^{\frac 12} \bigl[- \Delta_x + 0\bigr]^{-1} |v_{12}|^{\frac 12} \phi_0 = \phi_0, 
\end{equation}
see \cite{1} for details.  As we shall see (Remark~\ref{remark:3} below), the constant $C_0$ in (\ref{main}) can also be expressed through zero energy solutions of certain 
Birman-Schwinger operators.

The aim of this section is to prove 
\begin{theorem}\label{th:main2}
Suppose that the interaction between 
the particles $\{1,2\}$ is tuned to make them have a zero energy resonance 
and no negative energy bound states. Let $\psi_\lambda$ for $\lambda \in (\lambda_{cr}, \tilde \lambda)$ be the ground state wave function of $H(\lambda)$ 
defined in (\ref{hami}), which corresponds to the ground state energy $E(\lambda)$. 
Then there exists $C_0 \in (0, \infty) $ such that 
\begin{equation}\label{1.24:152}
\lim_{\lambda \to \lambda_{cr} + 0} \bigl|\ln |E(\lambda)|\bigr| \left\{ \bigl\||v_{13}|^{\frac12}\psi_\lambda \bigr\|^2 +  \bigl\||v_{23}|^{\frac12}
\psi_\lambda \bigr\|^2  \right\} = C_0 . 
\end{equation}
\end{theorem}
(The same constant $C_0$ is used in Theorem~\ref{th:main}). We shall extensively use the results from \cite{3}, therefore it is convenient 
to pass to sequences. Let $\lambda_n \in (\lambda_{cr}, {\tilde \lambda})$ be any sequence such that $\lambda_n \to \lambda_{cr}$ and $\psi_n\equiv \psi_{\lambda_n}$. 
Instead of (\ref{1.24:152}) it suffices to prove that 
\begin{equation}\label{1.24:15}
\lim_{n \to \infty} |\ln k_n| \left\{ \bigl\||v_{13}|^{\frac12}\psi_n\bigr\|^2 +  \bigl\||v_{23}|^{\frac12}
\psi_n\bigr\|^2  \right\} = C_0/2 ,  
\end{equation}
where $k_n := |E(\lambda_n)|^\frac 12$. We need only to prove that the limit on the lhs of (\ref{1.24:15}) exists and is positive. 
Note that all requirements R1-R3 in \cite{3} are satisfied and the sequence $\psi_n$ totally spreads (see Sec. Theorems~1,3 in \cite{1} for the proof). 

Before we proceed with the proof let us introduce additional notations. Let $\mathcal{F}_{12}$ and $\mathbb{P}_0$ denote the partial Fourier transform and the 
projection operator, which act on $f(x,y)$ as follows 
\begin{gather}
\hat f(x,p_y)  =  \mathcal{F}_{12} f(x,y) = \frac 1{(2 \pi )^{3/2}} \int d^3 y
\; e^{-ip_y \cdot \; y} f(x,y) ,\label{ay5} \\
[\mathbb{P}_0 f ] (x,y) = \phi_0 (x) \int f(x',y) \phi_0 (x') d^3 x' , 
\end{gather}
and where $\phi_0$ is defined in (\ref{1.16;1}). 
For a shorter notation let us denote 
\begin{equation}
M_n :=  \left\{ \bigl\||v_{13}|^{\frac12}\psi_n\bigr\|^2 +  \bigl\||v_{23}|^{\frac12}
\psi_n\bigr\|^2  \right\}^{1/2} . 
\end{equation}
Similar to Eqs. (39)-(40) in \cite{3} we introduce the operator function
\begin{equation}\label{b12}
 \tilde B_{12} (k_n) := \mathcal{F}^{-1}_{12} \xi_n (p_y)
\mathcal{F}_{12} ,
\end{equation}
where
\begin{equation}\label{tail}
    \xi_n (p_y) := \left\{ \begin{array}{ll}
    |p_y|^{\delta/8} + (k_n)^{\delta/8}  & \quad \mathrm{if} \;\; |p_y|
\leq 1  \\
    1 + (k_n)^{\delta/8} & \quad \mathrm{if} \;\; |p_y| \geq 1 . \\
    \end{array}
    \right.
\end{equation}
\begin{lemma}\label{lemma:1}
The sequences $\varphi^{(1)}_n := M_n^{-1}|v_{13}|^{\frac 12} \psi_n$ and $\varphi^{(2)}_n := M_n^{-1}|v_{23}|^{\frac 12} \psi_n$, where $\psi_n$ 
is defined in Theorem~\ref{th:main2}, converge in norm. 
The sequence $\varphi^{(3)}_n := M_n^{-1}\tilde B_{12} (k_n)|v_{12}|^{\frac 12} \psi_n$ is uniformly norm-bounded. 
\end{lemma}
\begin{proof}
From the Schr\"odinger equation for $\psi_n$ it follows that 
\begin{equation}
 \lambda_n^{-1} \left( \begin{array}{c}
\varphi^{(1)}_n\\
\varphi^{(2)}_n
\end{array}\right) 
=\mathcal{K} (k_n^2) \left( \begin{array}{c}
\varphi^{(1)}_n\\
\varphi^{(2)}_n 
\end{array} \right) , 
\end{equation}
where $\mathcal{K}(z)$ for $z>0$ is a bounded operator on $L^2 (\mathbb{R}^6) \oplus L^2 (\mathbb{R}^6)$ defined as 
\begin{equation}\label{blacblac}
  \mathcal{K}(z) :=  \left( \begin{array}{cc}
|v_{13}|^{\frac 12} \bigl( H_0 + v_{12} + z\bigr)^{-1} |v_{13}|^{\frac 12} &|v_{13}|^{\frac 12} \bigl( H_0 + v_{12} + z\bigr)^{-1} |v_{23}|^{\frac 12}\\
|v_{23}|^{\frac 12} \bigl( H_0 + v_{12} + z\bigr)^{-1} |v_{13}|^{\frac 12} &|v_{23}|^{\frac 12} \bigl( H_0 + v_{12} + z\bigr)^{-1} |v_{23}|^{\frac 12}  
\end{array} \right) .
\end{equation}
Like in Sec.~II in \cite{4} one proves that $\mathcal{K}(z)$ for $z\to +0$ has a norm limit $\mathcal{K}(0)$, besides $\mathcal{K}(z)$ for $z\geq 0$ is a positivity 
preserving self-adjoint operator. Let us show that the off-diagonal terms in (\ref{blacblac}) are compact operators on $L^2 (\mathbb{R}^6)$. By the resolvent identity 
\begin{gather}
 |v_{13}|^{\frac 12} \bigl( H_0 + v_{12} + z\bigr)^{-1} |v_{23}|^{\frac 12} =  |v_{13}|^{\frac 12} \bigl( H_0 + z\bigr)^{-1} |v_{23}|^{\frac 12} \nonumber\\
+ \left[|v_{13}|^{\frac 12} \bigl( H_0 + z\bigr)^{-1} |v_{12}|^{\frac 12} \right] |v_{12}|^{\frac 12}  \bigl( H_0 + v_{12} + z\bigr)^{-1} |v_{23}|^{\frac 12} = \left[ |v_{13}|^{\frac 12} \bigl( H_0 + z\bigr)^{-1} |v_{23}|^{\frac 12}  \right] \nonumber\\
+ \left[|v_{13}|^{\frac 12} \bigl( H_0 + z\bigr)^{-1} |v_{12}|^{\frac 12} \right] \left[|v_{12}|^{\frac 12}  \bigl( H_0 + z\bigr)^{-1} |v_{23}|^{\frac 12} \right] \nonumber\\
- \left[|v_{13}|^{\frac 12} \bigl( H_0 + z\bigr)^{-1} |v_{12}|^{\frac 12} \right] \left\{|v_{12}|^{\frac 12} \bigl( H_0 + v_{12}+z\bigr)^{-1}  |v_{12}|^{\frac 12}\right\}
\left[|v_{12}|^{\frac 12} \bigl( H_0 + z\bigr)^{-1} |v_{23}|^{\frac 12}  \right]. \label{gogol}
\end{gather}
On the rhs of (\ref{gogol}) all operators in square brackets are Hilbert-Schmidt (the result of this sort goes back at least to \cite{ginibre}, see also \cite{1} for the 
proof in present notations). 
Since the operator in curly brackets is bounded for $z>0$ the lhs of (\ref{gogol}) is a 
compact operator. Thus by Weyl's criterion \cite{teschl,reed} 
\begin{equation}
 \sigma_{ess} \bigl( \mathcal{K}(z) \bigr) = \sigma_{ess} \bigl( \mathcal{K}_{11}(z) \bigr) \cup \sigma_{ess} \bigl( \mathcal{K}_{22}(z) \bigr)
\subseteq [0, 1/{\tilde \lambda}] , 
\end{equation}
where ${\tilde \lambda}$ was defined in (\ref{tildelambda}). $(\varphi^{(1)}_n , 
\varphi^{(2)}_n )$ 
is a normalized eigenvector of $\mathcal{K}(k_n^2)$ corresponding to the eigenvalue $ \lambda_n^{-1}$. Due to the location of the essential spectrum $\|\mathcal{K}(k_n^2)\|$ 
equals the maximal eigenvalue of $\mathcal{K}(k_n^2)$. Because $\varphi^{(1)}_n , \varphi^{(2)}_n \geq 0$ we conclude due to the positivity preserving property that 
$\|\mathcal{K}(k_n^2)\| = \lambda_n^{-1}$ (see Theorem XIII.43 in Vol. 4 of \cite{reed}). Therefore, due to the norm convergence $\lambda_{cr}^{-1} = \|\mathcal{K}(0)\| $ is the maximal 
eigenvalue of $\mathcal{K}(0) $, which is isolated and non-degenerate. Let $(\varphi^{(1)}_\infty ,  \varphi^{(2)}_\infty )$ with $\varphi^{(1)}_\infty, \varphi^{(2)}_\infty \geq 0$ 
be the eigenvector of $\mathcal{K}(0) $, which corresponds to 
$\lambda_{cr}^{-1}$. Again by the norm convergence of $\mathcal{K}(k_n^2)$ we have that $\varphi^{(1)}_n \to \varphi^{(1)}_\infty$ and 
$\varphi^{(2)}_n \to \varphi^{(2)}_\infty$ in norm. 

To prove that $\sup_n \|\varphi^{(3)}_n\| < \infty$ note that by Eq. (67) in \cite{3} 
\begin{gather}
  \varphi^{(3)}_n =  \lambda_n \Bigl\{1- |v_{12}|^{\frac 12} \bigl(H_0 +
k_n^2\bigr)^{-1}|v_{12}|^{\frac 12} \Bigr\}^{-1}
|v_{12}|^{\frac 12} \tilde B_{12} (k_n) \bigl[H_0 + k_n^2\bigr]^{-1}|v_{13}|^{\frac 12}\varphi^{(1)}_n \nonumber\\
 + \lambda_n \Bigl\{1- |v_{12}|^{\frac 12} \bigl(H_0 +
k_n^2\bigr)^{-1}|v_{12}|^{\frac 12} \Bigr\}^{-1}
|v_{12}|^{\frac 12} \tilde B_{12} (k_n) \bigl[H_0 + k_n^2\bigr]^{-1} |v_{23}|^{\frac 12}\varphi^{(2)}_n . 
\end{gather}
Without loosing generality it suffices to prove that 
\begin{equation}
  \varphi^{(4)}_n := \Bigl\{1- |v_{12}|^{\frac 12} \bigl(H_0 +
k_n^2\bigr)^{-1}|v_{12}|^{\frac 12} \Bigr\}^{-1}
|v_{12}|^{\frac 12} \tilde B_{12} (k_n) \bigl[H_0 + k_n^2\bigr]^{-1}|v_{13}|^{\frac 12}\varphi^{(1)}_n \label{29.12:1}
\end{equation}
is uniformly norm-bounded. Denoting $\hat  \varphi^{(4)}_n \equiv \mathcal{F}_{12} \varphi^{(4)}_n $ we get 
\begin{equation}
 \hat  \varphi^{(4)}_n = \chi_{[0, \rho_0]} \Bigl( \sqrt{p_y^2 +k_n^2}\Bigr) \hat  \varphi^{(4)}_n + \chi_{(\rho_0 , \infty)} \Bigl( \sqrt{p_y^2 +k_n^2}\Bigr) \hat  \varphi^{(4)}_n = 
\chi_{[0, \rho_0]} \Bigl( \sqrt{p_y^2 +k_n^2}\Bigr) \hat  \varphi^{(4)}_n + \mathcal{O}(1), 
\end{equation}
where $\mathcal{O}(1)$ denotes the terms, which are uniformly norm-bounded. Here $\rho_0 >0$ is a fixed cutoff parameter. 
Following Lemma~11 in \cite{1} (the value of $\rho_0$ is also defined there) we can expand the operator in curly brackets in (\ref{29.12:1}). 
This expansion gives (see Eqs. (73)-(74) in \cite{3})
\begin{gather}
\chi_{[0, \rho_0]} \Bigl( \sqrt{p_y^2 +k_n^2}\Bigr) \hat  \varphi^{(4)}_n =\chi_{[0, \rho_0]} \Bigl( \sqrt{p_y^2 +k_n^2}\Bigr) 
a^{-1} \mathbb{P}_0 |v_{12}|^{\frac 12}  \nonumber\\
\times \bigl(|p_y|^2 + k_n^2\bigr)^{-\frac 12}\xi_n (p_y) \bigl[-\Delta_x + p_y^2 + k_n^2\bigr]^{-1}\widehat{|v_{13}|^{\frac 12}} \hat \varphi^{(1)}_n + \mathcal{O}(1) , 
\end{gather}
where $\widehat{|v_{13}|^{\frac 12}} := \mathcal{F}_{12} |v_{13}|^{\frac 12} \mathcal{F}^{-1}_{12}$. Thus 
\begin{equation}\label{1.24:12}
 \|\varphi^{(4)}_n\| \leq a^{-1} \left\| |v_{12}|^{\frac 12} \chi_{[0, \rho_0]} ( |p_y|)  \bigl(|p_y|^2 + k_n^2\bigr)^{-\frac 12} \xi_n (p_y) \bigl[-\Delta_x + p_y^2 + k_n^2\bigr]^{-1}\widehat{|v_{13}|^{\frac 12}}\right\| + \mathcal{O}(1)
\end{equation}
It remains to prove that the operator norm on the right hand side (rhs) of (\ref{1.24:12}) is uniformly bounded. This can be trivially estimated through the  
Hilbert-Schmidt norm (c. f. proof of Lemma~9 in \cite{1})
\begin{gather}
\left\| |v_{12}|^{\frac 12}  \chi_{[0, \rho_0]} ( |p_y|) \bigl(|p_y|^2 + k_n^2\bigr)^{-\frac 12} \xi_n (|p_y|) \bigl[-\Delta_x + p_y^2 + k_n^2\bigr]^{-1}\widehat{|v_{13}|^{\frac 12}}\right\|^2 \nonumber \\
\leq c \int_{|p_y|\leq \rho_0} d^3 p_y \frac{\xi_n^2 (|p_y|)}{( p_y^2 + k_n^2)^{3/2}} , \label{1.24:14} 
\end{gather}
where $c>0$ is a constant. The integral on the rhs of (\ref{1.24:14}) is clearly convergent and uniformly bounded. 
\end{proof}

\begin{proof}[Proof of Theorem~\ref{th:main2}]
Following \cite{3} let us denote 
\begin{equation}\label{Phi_n2}
\Phi_n^{(i)} := - \lambda_n \sqrt{|v_{12}|} \bigl[H_0 + k_n^2\bigr]^{-1} v_{i3} \psi_n  \quad\quad (i=1,2)
\end{equation}
and 
\begin{equation}\label{gn}
 g_n (y) = g_n^{(1)} (y) + g_n^{(2)} (y) , 
\end{equation}
where 
\begin{equation}\label{gni}
g_n^{(i)} (y) := \int d^3 x \phi_0 (x) \Phi_n^{(i)}  (x, y) .  
\end{equation}
The functions $g(y), g_n^{(i)} (y) \in L^1 (\mathbb{R}^3) \cap  L^2 (\mathbb{R}^3) $ coincide with the ones defined in Eqs. (93), (94) in \cite{3}.

In \cite{3} it is proved that 
\begin{equation}\label{1.16;3}
 |\hat g_n (0)||\ln k_n|^{1/2} \to \frac{\sqrt 2 a}{R(0)}  >0 , 
\end{equation}
where $\hat g_n (p_y)$ is the Fourier image of $g_n$ defined in (\ref{gn}) and 
\begin{equation}
R(0) = \int  \frac{\phi_0 (x')}{|x'|} \bigl|V_{12} (\alpha x')\bigr|^{\frac 12} . 
\end{equation}
(\ref{1.16;3}) follows from Eq.~(90) in \cite{3} and the fact that the norm of the function on the rhs of that equation goes to one for $n \to \infty$ 
(see the text under Eq.~(90) in \cite{3}). 

Our aim is to show that the sequence $|\hat g_n (0)| /M_n$ converges and $\lim_{n \to \infty} |\hat g_n (0)| /M_n = C_1 >0$; 
then (\ref{1.24:15}) follows with $C_0 = 4a^2 [C_1 R(0)]^{-2}$. Note that due to the positivity preserving property of $[H_0+k_n^2]^{-1}$ 
\begin{equation}\label{1.24:17}
 \frac{|\hat g_n (0)| }{M_n} = \frac{\|g^{(1)}_n\|_1 }{M_n} + \frac{\|g^{(2)}_n\|_1 }{M_n} , 
\end{equation}
where $g^{(1,2)} (y) \in L^1 (\mathbb{R}^3) \cap  L^2 (\mathbb{R}^3) $ are defined in Eqs. (93), (94) in \cite{3}. 
It suffices to prove the convergence of the first term on the rhs of (\ref{1.24:17}) (the convergence of the second term is proved analogously). 
Similar to \cite{3} let us introduce Jacobi coordinates $\eta, \zeta$, which are pictured in Fig.~\ref{fig:1} (right) 
\begin{gather}
 \eta = \alpha'^{-1} (r_3 - r_1) \nonumber\\
\zeta = \left[\frac{2(m_1 +m_3) m_2}{\hbar^2 (m_1 + m_2 + m_3)} \right]^{\frac 12} \left( r_2 - \frac{m_1}{m_1+m_3} r_1 - \frac{m_3}{m_1+m_3} r_3\right) \nonumber , 
\end{gather}
where $\alpha' = \hbar (m_1 + m_3)^{\frac 12} (2m_1 m_3)^{-\frac 12}$. 
The coordinates $(\eta , \zeta)$ and $(x,y)$ are connected through the orthogonal linear
transformation
\begin{gather}
 x = m_{x\eta} \eta +  m_{x\zeta} \zeta , \nonumber\\
 y = m_{y\eta} \eta +  m_{y\zeta} \zeta,\nonumber
\end{gather}
where $m_{x\eta}, m_{x\zeta} \neq 0, m_{y\eta}, m_{y\zeta}$
are real and can be expressed through mass ratios in the
system. For all $R >0$ 
we have 
\begin{gather}
 M_n^{-1} \left\| g^{(1)}_n\right\| _1 = \left\| \chi_{[0, R] }(|\eta|)\phi_0 \mathfrak{X}_n \varphi_n^{(1)}\right\| _1 \nonumber\\
+ 
M_n^{-1}\left\| \chi_{(R, \infty) }(|\eta|) \phi_0 |v_{12}|^{\frac 12} \left[H_0 + k_n^2 \right]^{-1} |v_{13}| \psi_n \right\| _1 , \label{1.24:21}
\end{gather}
where we have defined
\begin{equation}
\mathfrak{X}_n :=  |v_{12}|^{\frac 12} \left[H_0 + k_n^2 \right]^{-1}  |v_{13}|^{\frac 12} . 
\end{equation}
The operators 
$\mathfrak{X}_n : L^2 (\mathbb{R}^6) \to L^2 (\mathbb{R}^6) $ are norm-bounded and have a norm limit for $n \to \infty$, which 
we denote as $\mathfrak{X}_0$ (for the proof see f. e. Lemma~7 in \cite{1}). Thus by Lemma~\ref{lemma:1} $\mathfrak{X}_n \varphi_n^{(1)} \to \mathfrak{X}_0 \varphi_\infty^{(1)}$ in $L^2$ sense. Then 
$\chi_{[0, R] }(|\eta|)\phi_0 \mathfrak{X}_n \varphi_n^{(1)}$ converges to $\chi_{[0, R] }(|\eta|)\phi_0 \mathfrak{X}_0 \varphi_\infty^{(1)}$ in $L^1$ sense because by 
the Cauchy-Schwarz inequality 
\begin{equation}
  \left\|\chi_{[0, R] }(|\eta|)\phi_0 \bigl( \mathfrak{X}_n \varphi_n^{(1)} - \mathfrak{X}_0 \varphi_\infty^{(1)}\bigr)\right\|_1 \leq \|\chi_{[0, R] }(|\eta|)\phi_0 \| 
 \left\|\mathfrak{X}_n \varphi_n^{(1)} - \mathfrak{X}_0 \varphi_\infty^{(1)}\right\| , 
\end{equation}
where $\|\chi_{[0, R] }(|\eta|)\phi_0 \| $ is finite. Hence,  the first term on the rhs of (\ref{1.24:21}) converges for all $R >0$. 
Now the convergence of the sequence on the left hand side (lhs) of (\ref{1.24:21}) follows from Lemmas~\ref{lemma:2}, \ref{lemma:3}. 
We have proved that the sequence on the lhs of (\ref{1.24:17}) converges to $C_1 \in [0, \infty)$. 
It remains to show that 
$C_1 \neq 0$. This follows from the fact that $\varphi_\infty^{(1)}, \varphi_\infty^{(2)} \geq 0 $ and besides $\|\varphi_\infty^{(1)}\|^2 + \|\varphi_\infty^{(2)}\|^2 =1$, so 
at least one of 
the terms on the rhs of (\ref{1.24:17}) converges to a positive value. 
\end{proof}

\begin{remark}\label{remark:3}
 The proof of Theorem~\ref{th:main2} allows to express the constant $C_0$ in Theorem~\ref{th:main} in terms of zero energy solutions of Birman-Schwinger operators, namely, 
$C_0 = 4a^2 [R(0)C_1]^{-2}$, where 
\begin{equation}
C_1 =  \bigl\|\phi_0 \mathfrak{X}_0 \varphi_\infty^{(1)}\bigr\|_1 + \bigl\|\phi_0 \mathfrak{Y}_0 \varphi_\infty^{(2)}\bigr\|_1
\end{equation}
and by definition $\mathfrak{Y}_0$ is the operator norm limit 
\begin{equation}
 \mathfrak{Y}_0 =\lim_{z \to +0} |v_{12}|^{\frac 12} \left[H_0 + z \right]^{-1}  |v_{23}|^{\frac 12} . 
\end{equation}
\end{remark}

\begin{lemma}\label{lemma:2}
 For $R \to \infty$ 
\begin{equation}
 \sup_{n} M_n^{-1}\left\| \chi_{(R, \infty) }(|\eta|) \phi_0 |v_{12}|^{\frac 12} \left[H_0 + k_n^2 \right]^{-1} |v_{13}| \psi_n \right\| _1 \to 0 . 
\end{equation}
\end{lemma}
\begin{proof}
 The proof is partly based on that of Lemma~4 in \cite{3}, and we need to introduce additional notations from \cite{3}. 
$\mathcal{F}_{13}$ denotes the partial Fourier transform, which acts on $f(\eta,
\zeta) $ as
\begin{equation}\label{xw20}
 \mathcal{F}_{13} f := \hat f(\eta, p_\zeta) = \frac 1{(2 \pi )^{3/2}} \int d^3
\zeta \; e^{-ip_\zeta \cdot \; \zeta} f(\eta, \zeta) .
\end{equation}
Let us introduce the operator function
\begin{equation}\label{b13}
 \tilde B_{13} (k_n) := \mathcal{F}^{-1}_{13} \tilde t_n (p_\zeta)
\mathcal{F}_{13} ,
\end{equation}
where
\begin{equation}\label{tail2}
    \tilde t_n (p_\zeta) = \left\{ \begin{array}{ll}
    |p_\zeta|^{1-\delta} + (k_n)^{1-\delta}  & \quad \mathrm{if} \;\; |p_\zeta|
\leq 1  \\
    1 + (k_n)^{1-\delta} & \quad \mathrm{if} \;\; |p_\zeta| \geq 1 . \\
    \end{array}
    \right.
\end{equation}
Following Eq.~(30) in \cite{3} we define 
\begin{equation}
  F_n :=  \lambda_n \bigl[H_0 + k_n^2\bigr]^{-1} v_{13} \psi_n . 
\end{equation}
Then using Eq.~(98) in \cite{3} we obtain the inequality 
\begin{gather}
M_n^{-1}\left\| \chi_{(R, \infty) }(|\eta|) \phi_0 |v_{12}|^{\frac 12} \left[H_0 + k_n^2 \right]^{-1} |v_{13}| \psi_n \right\| _1 = 
M_n^{-1} \lambda_n^{-1}\left\| \chi_{(R, \infty) }(|\eta|) \phi_0 |v_{12}|^{\frac 12} F_n \right\| _1 \nonumber\\
\leq  M_n^{-1} \lambda_n^{-1} \sum_{i=1}^3\left\| \chi_{(R, \infty) }(|\eta|) \phi_0 |v_{12}|^{\frac 12} \tilde F_n^{(i)} \right\| _1  , \label{1.20;23}
\end{gather}
where 
\begin{equation}
 \tilde F^{(i)}_n := \bigl[ H_0 + k_n^2\bigr]^{-1} |v_{13}|^{1/2}
\tilde B_{13}(k_n) \Psi^{(i)}_n   \label{xw17}
\end{equation}
and 
\begin{gather}
  \Psi^{(1)}_n  := |v_{13}|^{1/2} \tilde B_{13}^{-1} (k_n) \bigl[ H_0 +
k_n^2\bigr]^{-1} |v_{12}| |\psi_n| , \label{psin1}\\
 \Psi^{(2)}_n  := \lambda_n |v_{13}|^{1/2} \tilde B_{13}^{-1} (k_n) \bigl[ H_0 +
k_n^2\bigr]^{-1} |v_{23}| |\psi_n|  \label{psin2},\\
 \Psi^{(3)}_n  := \lambda_n  |v_{13}|^{1/2} \bigl[ H_0 + k_n^2\bigr]^{-1}  |v_{13}|^{1/2}
  \Bigl\{1 - \lambda_n |v_{13}|^{1/2} \bigl[H_0 
+ k_n^2\bigr]^{-1} |v_{13}|^{1/2} \Bigr\}^{-1} \nonumber \\\times\tilde B_{13}^{-1} (k_n)   |v_{13}|^{1/2} 
 \bigl[ H_0 + k_n^2\bigr]^{-1}  \bigl(|v_{12}| +  \lambda_n
|v_{23}|\bigr)|\psi_n| . \label{psin3}
\end{gather}
Eqs. (\ref{xw17}) and (\ref{psin1})-((\ref{psin3})) come from Eqs. (43), (44), (48) and (49) in \cite{3}, where one has to take into 
account that all interaction potentials 
are nonpositive. By Eq. (103) in \cite{3}
\begin{equation}\label{xw23}
\bigl| \tilde F^{(i)}_n (\eta , \zeta) \bigr| \leq \frac 1{2^{7/2} \pi^{5/2}}  \int d^3
\eta' \int d^3 p_\zeta \; \bigl| V_{13}(\alpha' \eta')\bigr|^{1/2}
\frac{e^{-\sqrt{p_\zeta^2 + k_n^2}|\eta-\eta'|}}{|\eta-\eta'|} \tilde t_n
(p_\zeta) \bigl| \hat \Psi^{(i)}_n (\eta', p_\zeta) \bigr| , 
\end{equation}
where $\hat \Psi^{(i)}_n  = \mathcal{F}_{13}  \Psi^{(i)}_n $. Using the exponential falloff of $V_{12} (x)$ from (\ref{1.16;1}) one can easily derive the inequality
\begin{equation}\label{20.1;15}
 \phi_0 (x) \bigl|V_{12} (\alpha x) \bigr| \leq \tilde b_1 e^{-\tilde b_2 |x|}, 
\end{equation}
where $\tilde b_{1,2}$ are constants. After substituting (\ref{20.1;15}) and (\ref{xw23}) into (\ref{1.20;23}), interchanging the order of integration and applying the Cauchy-Shwarz 
inequality we get 
\begin{gather}
M_n^{-1}\left\| \chi_{(R, \infty) }(|\eta|) \phi_0 |v_{12}|^{\frac 12} \left[H_0 + k_n^2 \right]^{-1} |v_{13}| \psi_n \right\| _1 \nonumber\\
 \leq 
\frac {\tilde b_1 \lambda_n^{-1}}{2^{7/2} \pi^{5/2}}
\sum_{i=1}^3  M_n^{-1} \bigl\| \Psi^{(i)}_n \bigr\|
\left\{ \int d^3 \eta' \int d^3 p_\zeta \; \bigl| V_{13}(\alpha' \eta')\bigr| \,
{\tilde t}^{\; 2}_n (p_\zeta) \, {\tilde J}^2 (\eta', p_\zeta) \right\}^{1/2} , 
\end{gather}
where 
\begin{equation}
 {\tilde J} (\eta', p_\zeta):= \int_{|\eta| > R} d^3 \eta \int d^3 \zeta  \;
 \frac{e^{-\sqrt{p_\zeta^2 + k_n^2}|\eta-\eta'|}}{|\eta-\eta'|}  e^{-\tilde b_2 |m_{x\eta} \eta +
m_{x\zeta} \zeta|} . 
\end{equation}
By Lemma~\ref{lemma:4} we only need to prove that $ \sup_{n} I_n \to 0 $ for $R \to \infty$, where we defined
\begin{equation}
I_n := \int d^3 \eta' \int d^3 p_\zeta \; \bigl| V_{13}(\alpha' \eta')\bigr| \,
{\tilde t}^{\; 2}_n (p_\zeta) \, {\tilde J}^2 (\eta', p_\zeta) . 
\end{equation}
Let us split the last integral as $I_n = I_n^{(1)} + I_n^{(2)}$, where
\begin{gather}
I_n^{(1)} := \int_{|\eta'| \leq R/2} d^3 \eta' \int d^3 p_\zeta \; \bigl| V_{13}(\alpha' \eta')\bigr| \,
{\tilde t}^{\; 2}_n (p_\zeta) \, {\tilde J}^2 (\eta', p_\zeta) , \label{1.24:31}\\
I_n^{(2)} := \int_{|\eta'|>R/2} d^3 \eta' \int d^3 p_\zeta \; \bigl| V_{13}(\alpha' \eta')\bigr| \,
{\tilde t}^{\; 2}_n (p_\zeta) \, {\tilde J}^2 (\eta', p_\zeta) . \label{1.24:32}
\end{gather}
Clearly, we can write 
\begin{equation}\label{1.24:37}
 {\tilde J} (\eta', p_\zeta) \leq c_1 \int_{|\eta| >R} d^3 \eta \frac{e^{-\sqrt{p_\zeta^2 + k_n^2}|\eta-\eta'|}}{|\eta-\eta'|} , 
\end{equation}
where $c_1 >0$ is a constant. For $|\eta'| \leq R/2$ and $|\eta| > R/2$ we have $|\eta - \eta'| \geq |\eta| - |\eta'| > |\eta|/2$, which gives the estimate
\begin{equation}
 {\tilde J} (\eta', p_\zeta) \leq c_2 \frac{e^{-\sqrt{p_\zeta^2 + k_n^2}R/2}}{\sqrt{p_\zeta^2 + k_n^2}} 
\left( R + \frac2{\sqrt{p_\zeta^2 + k_n^2}} \right) , 
\end{equation}
where $c_2 >0$ is a constant. Substituting this estimate into (\ref{1.24:31}) and using that $V_{13} (x) \in L^1 (\mathbb{R}^3)$ we obtain 
\begin{equation}
 I_n^{(1)} \leq c_3 \int d^3 p_\zeta \, {\tilde t}^{\; 2}_n (p_\zeta) \frac{e^{-\sqrt{p_\zeta^2 + k_n^2}R}}{p_\zeta^2 + k_n^2} \left[R^2 + \frac4{p_\zeta^2 + k_n^2}\right], 
\end{equation}
where $c_3 >0$ is a constant. Substituting Eq. (40) from \cite{3} we get 
\begin{equation}
 I_n^{(1)} \leq c_4 \int_0^1 s^{2-2\delta} e^{-sR} \left[R^2 + \frac4{s^2} \right] ds
+ c_4 \int_1^\infty  e^{-sR} \left[R^2 + \frac4{s^2} \right]ds 
\end{equation}
where $c_4 >0$ is another constant. Let us set $A_\beta := \sup_{x\geq 0} x^{\beta} e^{-x}$, where $A_\beta$ is finite and depends only on $\beta$. Then 
\begin{equation}\label{1.24:34}
   I_n^{(1)} \leq c_4 R^{-\delta} (A_{2+\delta} + 4A_\delta) \int_0^1 s^{-3\delta} ds + c_4 R^2 e^{-R/2} \int_1^\infty e^{-sR/2}\bigl[1+ 4s^{-2}R^{-2}\bigr] ds = \hbox{o}(R)
\end{equation}

It is easy to see that the terms on the rhs of (\ref{1.24:34}) vanish for $R \to \infty$. Let us estimate $I_n^{(2)} $. From (\ref{1.24:37}) we get 
\begin{equation}
  {\tilde J} (\eta', p_\zeta) \leq c_1 \int d^3 \eta \frac{e^{-\sqrt{p_\zeta^2 + k_n^2}|\eta-\eta'|}}{|\eta-\eta'|} \leq \frac{c'_1}{p_\zeta^2} , 
\end{equation}
where $c'_1 >0$ is a constant. Substituting this into (\ref{1.24:32}) we obtain
\begin{equation}
 I_n^{(2)} \leq c'_1 \left\{\int_{|\eta'|>R/2} d^3 \eta' \bigl| V_{13}(\alpha' \eta')\bigr| \right\} \int d^3 p_\zeta \; {\tilde t}^{\; 2}_n (p_\zeta) |p_\zeta|^{-4} = \hbox{o}(R), 
\end{equation}
where the integral in curly brackets goes to zero because $V_{13} (x) \in L^1 (\mathbb{R}^3)$ and the second integral is uniformly bounded for all $n$. 
\end{proof}

\begin{lemma}\label{lemma:3}
 Suppose that the sequence $a_n \in \mathbb{C}$ is such that $a_n = b_n (R) + c_n (R)$, where $ b_n (R), c_n (R) \in \mathbb{C}$ depend on a parameter $R >0$. Additionally 
assume that $ b_n (R)$ is convergent for all $R >0$ and $\limsup_{n \to \infty} |c_n ( R)| \to 0$ for $R \to \infty$. Then $a_n$ converges. 
\end{lemma}
\begin{proof}
 The proof is a trivial application of the Cauchy convergence criterion. For any $\varepsilon >0$ fix $N_1, R >0$ so that $|c_n ( R)| < \varepsilon/3$ for all $n > N_1$. Choose $N_2$ 
so that $|b_n (R) - b_m (R)| < \varepsilon/3$ for all $n, m > N_2$. Then 
\begin{equation}
 |a_n - a_m| \leq |b_n (R) - b_m (R)| + |c_n (R)| + |c_m (R)| < \varepsilon
\end{equation}
for all $n, m > \max (N_1, N_2)$. 
\end{proof}

\begin{lemma}\label{lemma:4}
$\sup_n M_n^{-1} \|\Psi^{(i)}_n\| < \infty $, where $\Psi^{(i)}_n$ for $i=1,2,3$ are defined in Eqs.(\ref{psin1})-(\ref{psin3}). 
\end{lemma}
\begin{proof}
 We have $M_n^{-1} \Psi^{(2)}_n = \lambda_n \mathcal{T}^{(2)}_n \varphi_n^{(1)}$, where 
\begin{equation}
  \mathcal{T}^{(2)}_n := \tilde B_{13}^{-1} (k_n) |v_{13}|^{1/2}_- \bigl[ H_0 + k_n^2\bigr]^{-1} |v_{12}|^{1/2} 
\end{equation}
 are uniformly norm-bounded 
operators (they are defined in the same way as in Eq.~(63) in \cite{3}). Thus by Lemma~\ref{lemma:1} $\sup_n M_n^{-1} \|\Psi^{(2)}_n \| < \infty$. From definition of $\Psi^{(1)}_n$ we have 
\begin{equation}\label{1.24:41}
 M_n^{-1} \Psi^{(1)}_n = \mathfrak{T}_n \mathfrak{D}^*_n \varphi_n^{(3)} , 
\end{equation}
where 
\begin{gather}
\mathfrak{T}_n  := |v_{13}|^\frac12 \tilde  B_{12}^{-1} \bigl[ H_0 + k_n^2\bigr]^{\frac{-3+\delta}4}\\
\mathfrak{D}_n:= |v_{12}|^\frac12 \tilde B_{13}^{-1} (k_n) \bigl[ H_0 + k_n^2\bigr]^{\frac{-1-\delta}4}. 
\end{gather}
$\mathfrak{T}_n , \mathfrak{D}_n $ are bounded operators on $L^2 (\mathbb{R}^6)$. 
From (\ref{1.24:41}) $\sup_n M_n^{-1} \| \Psi^{(1)}_n\| < \infty $ follows from 
$\mathfrak{T}_n , \mathfrak{D}_n $ being uniformly norm-bounded. Let us start with estimating the norm of $\mathfrak{T}_n$. 
Let us define operator functions $\mathfrak{T}^{(1)}_n, \mathfrak{T}^{(2)}_n : \mathbb{R}^3 \to \mathfrak{B}\bigl(L^2 (\mathbb{R}^3)\bigr)$, 
which act on $h(\eta) \in L^2 (\mathbb{R}^3)$ as 
follows
\begin{gather}
  \mathfrak{T}^{(1)}_n (p_\zeta) h = |V_{13}(\alpha'\eta)|^{\frac 12}\left[\xi_n \bigl(m_{y\eta}(-i\nabla_\eta) + m_{y\zeta}p_\zeta\bigr) \right]^{-1}
\bigl[-\Delta_\eta +p_\zeta^2 + k_n^2\bigr]^{\frac{-3+\delta}4}  \nonumber \\
\times \chi_{[0,1]}\bigl(|-i\nabla_\eta|\bigr) h , \\
  \mathfrak{T}^{(2)}_n (p_\zeta) h = |V_{13}(\alpha'\eta)|^{\frac 12}\left[\xi_n \bigl(m_{y\eta}(-i\nabla_\eta) + m_{y\zeta}p_\zeta\bigr) \right]^{-1}
\bigl[-\Delta_\eta +p_\zeta^2 + k_n^2\bigr]^{\frac{-1-\delta}4} \nonumber \\
\times\chi_{(1,\infty)}\bigl(|-i\nabla_\eta|\bigr) h . 
\end{gather}
The operators like $ \chi_{[0,1]}\bigl(|-i\nabla_\eta|\bigr)$ act in the sense described in Chapter 4 in \cite{traceideals}. It is easy to see that 
\begin{equation}
 [\mathcal{F}_{13} \mathfrak{T}_n f] (\eta, p_\zeta) = \mathfrak{T}^{(1)}_n (p_\zeta) \hat f (\eta, p_\zeta) + \mathfrak{T}^{(2)}_n (p_\zeta) \hat f (\eta, p_\zeta)  , 
\end{equation}
where $\hat f (\eta, p_\zeta) \equiv \mathcal{F}_{13} f$. Thus 
\begin{equation}\label{1.24:44}
 \|\mathfrak{T}_n\| \leq \sup_{p_\zeta}\|\mathfrak{T}^{(1)}_n (p_\zeta)\| + \sup_{p_\zeta}\|\mathfrak{T}^{(2)}_n (p_\zeta)\| . 
\end{equation}
The operator norms on the rhs of (\ref{1.24:44}) can be bounded by the trace ideals norms, which in turn can be bounded by Theorem~4.1 in \cite{traceideals}. 
\begin{equation}
 \|\mathfrak{T}^{(1)}_n (p_\zeta)\| \leq \|\mathfrak{T}^{(1)}_n (p_\zeta)\|_2 \leq (2\pi\alpha')^{-\frac 32} \|V_{13}\|_1 [J_n^{(1)}(p_\zeta)]^{\frac 12}, 
\end{equation}
where 
\begin{gather}
 J_n^{(1)}(p_\zeta) := 
\int_{|s|\leq 1 } \frac{d^3 s}{\xi_n^2 (m_{y\eta} s + m_{y\zeta}p_\zeta) \bigl[s^2 + p_\zeta^2 + k_n^2\bigr]^{\frac{3-\delta}2}}\leq 
\int_{|s|\leq 1 }  \frac{d^3 s}{\bigl| m_{y\eta} s + m_{y\zeta}p_\zeta\bigr|^{\frac \delta4} |s|^{3-\delta}} \nonumber\\
 + \int_{|s|\leq 1 }  \frac{d^3 s}{|s|^{3-\delta}}  \leq \left[\int_{|s|\leq 1 } \frac{d^3 s}{|s|^{3-\delta/2}}\right]^{\frac{3-\delta}{3-\delta/2}}
\left[\int_{|s|\leq 1 }  \frac{d^3 s}{\bigl| m_{y\eta} s + m_{y\zeta}p_\zeta\bigr|^{\frac 32 - \frac \delta4} }\right]^{\frac{\delta}{6-\delta}} + \int_{|s|\leq 1 }  \frac{d^3 s}{|s|^{3-\delta}} \nonumber\\
\leq \left[\int_{|s|\leq 1 } \frac{d^3 s}{|s|^{3-\delta/2}}\right]^{\frac{3-\delta}{3-\delta/2}}
\left[\int_{|s|\leq 1 }  \frac{d^3 s}{m_{y\eta}|  s|^{\frac 32 - \frac \delta4} }\right]^{\frac{\delta}{6-\delta}} + \int_{|s|\leq 1 }  \frac{d^3 s}{|s|^{3-\delta}} . \label{1.24:46}
\end{gather}
In (\ref{1.24:46}) we have used H\"olders inequality. The integrals on the rhs of (\ref{1.24:46}) are convergent and independent of $p_\zeta$ and $n$, hence, $\sup_{p_\zeta}\|\mathfrak{T}^{(1)}_n (p_\zeta)\| < \infty$. Similarly 
\begin{equation}
 \|\mathfrak{T}^{(2)}_n (p_\zeta)\| \leq \|\mathfrak{T}^{(2)}_n (p_\zeta)\|_3 \leq (2\pi\alpha')^{-1} \|V_{13}\|_{3/2} [J_n^{(2)}(p_\zeta)]^{\frac 13}, 
\end{equation}
where 
\begin{gather}
 J_n^{(2)}(p_\zeta) := 
\int_{|s|> 1 } \frac{d^3 s}{\xi_n^3 (m_{y\eta} s + m_{y\zeta}p_\zeta) |s|^{\frac{9-3\delta}{2}}}\leq 
\int_{|s|> 1 }  \frac{d^3 s}{\bigl| m_{y\eta} s + m_{y\zeta}p_\zeta\bigr|^{\frac{3\delta}{8}} |s|^{\frac{9-3\delta}{2}}} \nonumber\\
+ \int_{|s|> 1 }  \frac{d^3 s}{|s|^{\frac{9-3\delta}{2}}}  \leq \frac 1{m_{y\eta}^3}\int_{|s'|\leq 1} \frac{d^3 s'}{|s'|^{\frac{3\delta}{8}}} 
+ 2 \int_{|s|> 1 }  \frac{d^3 s}{|s|^{\frac{9-3\delta}{2}}} . \label{1.24:48}
\end{gather}
The integrals on the rhs of (\ref{1.24:48}) obviously converge. Thus we find that $\sup_{p_\zeta}\|\mathfrak{T}^{(2)}_n (p_\zeta)\| < \infty$ and $\|\mathfrak{T}_n\|$ is uniformly bounded 
(note that $\|V_{13}\|_{3/2}$ in (\ref{1.24:48}) is bounded because $V_{13} \in L^1 (\mathbb{R}^3) \cap L^2 (\mathbb{R}^3)$). 

Now we pass to estimating  $\|\mathfrak{D}_n\|$ and use the same method. Calculations similar to above ones give 
\begin{gather}
  \|\mathfrak{D}_n\| \leq (2\pi\alpha)^{-\frac 32} \|V_{12}\|_1 \left[\sup_{p_y} J_n^{(3)}(p_y)\right]^{\frac 12} \nonumber\\
+ 
\sup_{p_y} \left\| \bigl|V_{12}(\alpha x)\bigr|^{\frac 12} \left[\tilde t_n \bigl( m_{x\zeta} (-i\nabla_x) + m_{y\zeta} p_y \bigr)\right]^{-1}\right\| , \label{1.24:53}
\end{gather}
where 
\begin{gather}
  J_n^{(3)}(p_y) := \int_{|s|\leq 1 } \frac{d^3 s}{{\tilde t_n}^2 (m_{x\zeta} s + m_{y\zeta}p_y) |s|^{1+\delta}} \leq 
\int_{|s|\leq 1 } \frac{d^3 s}{\bigl| m_{x\zeta} s + m_{y\zeta}p_y \bigr|^{2-2\delta} |s|^{1+\delta}} \nonumber\\
 + \int_{|s|\leq 1 } \frac{d^3 s}{|s|^{1+\delta}} \leq \left[\int_{|s|\leq 1 } \frac{d^3 s}{\bigl| m_{x\zeta} s + m_{y\zeta}p_y \bigr|^{3-\delta}}\right]^{\frac{2-2\delta}{3-\delta}} \left[\int_{|s|\leq 1 } \frac{d^3 s}{|s|^{3-\delta}}\right]^{\frac{1+\delta}{3-\delta}} + \int_{|s|\leq 1 } \frac{d^3 s}{|s|^{1+\delta}} \nonumber\\
 \leq \left[\int_{|s|\leq 1 } \frac{d^3 s}{m_{x\zeta} | s |^{3-\delta}}\right]^{\frac{2-2\delta}{3-\delta}} \left[\int_{|s|\leq 1 } \frac{d^3 s}{|s|^{3-\delta}}\right]^{\frac{1+\delta}{3-\delta}} + \int_{|s|\leq 1 } \frac{d^3 s}{|s|^{1+\delta}} . \label{1.24:51}
\end{gather}
The integrals on the rhs of (\ref{1.24:51}) converge and it remains to estimate the operator norm in (\ref{1.24:53}). 
\begin{gather}
\left\| \bigl|V_{12}(\alpha x)\bigr|^{\frac 12} \left[\tilde t_n \bigl( m_{x\zeta} (-i\nabla_x) + m_{y\zeta} p_y \bigr)\right]^{-1}\right\| \nonumber\\
 \leq 
\left\| \bigl|V_{12}(\alpha x)\bigr|^{\frac 12} \left[\tilde t_n \bigl( m_{x\zeta} (-i\nabla_x) + m_{y\zeta} p_y \bigr)\right]^{-1} \chi_{[0,1]} \left( \bigl| m_{x\zeta} (-i\nabla_x) + m_{y\zeta} p_y \bigr| \right)\right\|_2 + \bigl\| V_{12}\bigr\|_\infty \nonumber\\
 \leq \bigl\| V_{12}\bigr\|_\infty + (2\pi\alpha)^{-\frac 32} \|V_{12}\|_1 \left[\sup_{p_y} \int_{|m_{x\zeta}s + m_{y\zeta}p_y| \leq 1 } \frac{d^3 s}{\bigl|m_{x\zeta}s + m_{y\zeta}p_y\bigr|^{2-2\delta}}\right]^{\frac 12} \nonumber\\
\leq \bigl\| V_{12}\bigr\|_\infty + (2\pi\alpha m_{x\zeta})^{-\frac 32} \|V_{12}\|_1  \left[\int_{|s| \leq 1 } \frac{d^3 s}{|s|^{2-2\delta}} \right]^{\frac 12} , 
\end{gather}
where we have again used Theorem 4.1 in \cite{traceideals}. 
Thus we find that $\sup_n \|\mathfrak{D}_n\| < \infty$ and, hence, $\sup_n M_n^{-1} \|\Psi^{(1)}_n \| < \infty$.  
Note that because all potentials are nonpositive we have 
\begin{equation}
 M_n^{-1} \Psi^{(3)}_n = \lambda_n \mathcal{T}_n^{(1)} Q_n \left[M_n^{-1}\Psi^{(1)}_n + M_n^{-1}\Psi^{(2)}_n\right] , 
\end{equation}
where 
\begin{gather}
 \mathcal{T}^{(1)}_n := |v_{13}|^{1/2} \bigl[ H_0 + k_n^2\bigr]^{-1}  |v_{13}|^{1/2}  , \\
Q_n := \Bigl\{1 - \lambda_n |v_{13}|^{1/2} \bigl[H_0 
+ k_n^2\bigr]^{-1} |v_{13}|^{1/2} \Bigr\}^{-1} . 
\end{gather}
Operators $ \mathcal{T}^{(1)}_n , Q_n$ are defined in the same way in Eqs. (62) and (46) in \cite{3}. In \cite{3} 
it was proved that $\sup_n \|\mathcal{T}_n^{(1)}\| < \infty$ and 
$\sup_n \|Q_n\| < \infty$. Thus $\sup_n M_n^{-1} \| \Psi^{(3)}_n\| < \infty$ as claimed. 
\end{proof}

\end{document}